\newcommand{\bibpunct}[6]{}

\documentclass{amsproc}

\usepackage[initials]{amsrefs}

\usepackage{graphicx,caption}

\usepackage{amssymb}

\newcommand{\inprod}{\mathbin{\lrcorner}}

\usepackage{relsize}

\numberwithin{equation}{section}

\theoremstyle{plain}

\newtheorem{corollary}[subsection]{Corollary}
\newtheorem{lemma}[subsection]{Lemma}
\newtheorem{proposition}[subsection]{Proposition}
\newtheorem{theorem}[subsection]{Theorem}

\theoremstyle{definition}

\newtheorem{definition}[subsection]{Definition}

\usepackage[shortlabels]{enumitem}

\usepackage{array}

\usepackage{tikz-cd}

\newcommand{\M}{\mathcal{M}}

\newcommand{\CO}{\mathcal{O}}
\newcommand{\CU}{\mathcal{U}}

\DeclareMathOperator{\gh}{gh}
\DeclareMathOperator{\pa}{p}
\DeclareMathOperator{\Ber}{Ber}

\newcommand{\Z}{\mathbb{Z}}
\newcommand{\C}{\mathbb{C}}
\newcommand{\R}{\mathbb{R}}
\newcommand{\p}{\partial}
\newcommand{\half}{\tfrac{1}{2}}
\newcommand{\Om}{\Omega}
\newcommand{\om}{\omega}
\newcommand{\bull}{\bullet}

\renewcommand{\o}{\otimes}
\newcommand{\ho}{\mathbin{\widehat{\otimes}}}
\renewcommand{\[}{[\![}

\newcommand{\Id}{\text{Id}}
\newcommand{\vv}{\mathbf{v}}

\newcommand{\x}{\mathbf{x}}

\newcommand{\tintL}{{\textstyle\int_L}}
\newcommand{\PiSP}{\Pi\text{SP}}
\DeclareMathOperator{\End}{End}
\DeclareMathOperator{\GL}{GL}

\newcommand{\QQ}{\mathcal{Q}}
\renewcommand{\H}{\mathbf{H}}

\DeclareMathOperator{\Tot}{Tot}

\newcommand{\TW}{{\text{TW}}}

\renewcommand{\t}{\mathbf{t}}

\newcommand{\T}{\mathsf{Z}}

\renewcommand{\d}{\mathbf{d}}

\DeclareMathOperator{\so}{\mathbf{so}}

\newcommand{\g}{\mathbf{g}}

\newcommand{\Ss}{\mathbb{S}}
\renewcommand{\ss}{\mathsf{S}}
\newcommand{\TT}{\mathsf{T}}
\DeclareMathOperator{\Spin}{Spin}
\DeclareMathOperator{\SO}{SO}
\DeclareMathOperator{\SOSp}{SOSp}

\newcommand{\m}{\mathbf{m}}
\newcommand{\n}{\mathbf{n}}

\DeclareMathOperator{\cl}{c}
\newcommand{\tint}{{\textstyle\int}}

\title{Global gauge conditions in the Batalin--Vilkovisky formalism}

\author[E. Getzler and S. W. Pohorence]{Ezra Getzler and Sean Weinz
  Pohorence}

\address{Northwestern University, Evanston, Illinois, USA}


\thanks{The research of the first author was supported in part by a
  Simons Foundation Collaboration Grant. The research of the second
  author is supported in part by Research Training Grant ``Analysis on
  manifolds'' of the National Science Foundation at Northwestern
  University. We thank Chris Hull for introducing us to this subject,
  and alerting us to Siegel's work on this problem, and Andrei
  Mikhailov for pointing out the importance of allowing isotopies of
  Lagrangians in Corollary~5.3.}

\begin{document}

\begin{abstract}
  In the Batalin--Vilkovisky formalism, gauge conditions are expressed
  as Lagrangian submanifolds in the space of fields and antifields. We
  discuss a way of patching together gauge conditions over different
  parts of the space of fields, and apply this method to extend the
  light-cone gauge for the superparticle to a conic neighbourhood of
  the forward light-cone in momentum space.
\end{abstract}

\maketitle

The Gribov ambiguity concerns the impossibility, for topological
reasons, of making a global choice of gauge in non-Abelian Yang-Mills
theories \cites{Gribov,Singer}. An analogous phenomenon occurs in
string theory, where the light-cone gauge only yields a gauge
condition on a dense open subset of the forward light-cone in momentum
space. In this article, we introduce a technique for patching together
gauge conditions over different parts of the space of fields. In the
final section, we apply this method to give an extension of the
light-cone gauge to all of the forward light-cone; the usual
light-cone gauge only makes sense over an open dense subset of the
forward light-cone.

We work in the Batalin--Vilkovisky formalism, in which gauge
conditions are expressed by means of Lagrangian submanifolds of the
superspace of fields and antifields. Our main result shows how the
gauge conditions associated to Lagrangian submanifolds may be glued
together. In mathematics, this is called descent, and is typically
peformed using the language of simplicial manifolds and of
cosimplicial algebras. This is the approach we adopt here. We call the
resulting collection of Lagrangian submanifolds together with families
of isotopies between them \textbf{flexible Lagrangian submanifolds}.

The main ingredient in our construction is a homotopy formula for the
gauge conditions associated to a family of Lagrangian submanifolds,
due to Mikhailov and Schwarz \cite{MS}. Imitating Weil's proof of the
de Rham theorem (see \cite{BottTu}), we obtain our formula for
integration of half-forms over flexible Lagrangian submanifolds.

The BRS formalism is the special case of the Batalin--Vilkovisky
formalism in which the action has a linear dependence on the
antifields. (This is what Batalin and Vilkovisky call a rank one
theory.) In this case, our results specialize to the formula of
Becchi, Giusto and Imbimbo \cite{BGI}. Costello \cite{Costello} also
considers the problem of globalization of locally defined gauge
conditions. Fields are sections of a fiber bundle over the world-line,
and our results concern sheaves over the total space of this fiber
bundle, and not over the world-line, as in \cites{BGI,Costello}
(although our gauge conditions happen to be translation-invariant in
the world-line).

We apply this technique to the problem of gauge-fixing the
superparticle in the Batalin--Vilkovisky formalism; this case is not
covered by the BRS formalism, since the action has quadratic
dependence on the antifields. (In the terminology of Batalin and
Vilkovisky, this is a rank two theory.) We construct a flexible
Lagrangian submanifold giving a gauge condition in an open conic
neighbourhood of the forward light-cone in momentum space, based on a
pair of open sets $U_+$ and $U_-$ contained in the region $\{p_0>0\}$,
and gauge conditions determined by the light-like vectors
$(1,0,\dots,0,\pm1)\in \R^{1,9}$. There are Lagrangian submanifolds
$L_\pm$ lying over $U_\pm$, and a Lagrangian isotopy between them over
the intersection
\begin{equation*}
  U_{+-} = U_+\cap U_- .
\end{equation*}

This flexible Lagrangian submanifold is neither Lorentz invariant in
the target superspace nor invariant under supersymmetry. In order to
show that our functional integral respects these symmetries, we extend
our formula to the case where the theory carries an action of a
(finite-dimensional) Lie superalgebra. As in the BRST formalism, this
formula takes values in the differential graded commutative algebra of
Lie superalgebra cochains, and exhibits the equivariance of our
construction with respect to this Lie superalgebra. Applied to the
superparticle, this shows that our gauge-fixed functional integral is
Lorentz invariant and supersymmetric.

\section*{Convention}

Throughout this paper, we use the terms submanifold and subspace when
referring to sub(super)manifolds and sub(super)spaces. Unless
otherwise mentioned, we make use of the Einstein summation convention.

\section{Odd symplectic structures}

In this section, we review several results on the supergroup of
symmetries of an odd symplectic superspace. Our references are Manin
\cite{Manin}*{Section 3.5} and Khudaverdian and Voronov \cite{KV}.

Let $V$ be a finite-dimensional superspace, with homogeneous basis
$\{e_a\}_{a\in I}$: each vector $e_a$ has a well-defined parity
$p_a=\pa(e_a)$. The dual superspace $V^*$ has the dual basis
$\{e^*_a\}_{a\in I}$, where $e_a^*$ has the same parity as $e_a$, and
\begin{equation*}
  e_a^*(e_b) = \delta_{ab} .
\end{equation*}
The sign rule implies that the basis $\{e_a^{**}\}_{a\in I}$ of the
double dual $V^{**}$ differs from the original basis
$\{e_a\}_{a\in I}$ by a sign:
\begin{equation*}
  e_a^{**} = (-1)^{p_a} e_a .
\end{equation*}

Let $W$ be a second finite-dimensional superspace, with homogeneous
basis $\{f_b\}_{b\in J}$, with parity $q_b=\pa(f_b)$. Given a morphism
$A:V\to W$, we define the matrix elements of $A$ by
\begin{equation*}
  A_a^b = f_b^*(Ae_a) .
\end{equation*}
For a general vector $\vv=v^ae_a$, the morphism $A$ acts by the
formula
\begin{equation*}
  A\vv = (A^b_av^a) f_b .
\end{equation*}

The supertranspose $A^*:W^*\to V^*$ of a morphism $A:V\to W$ is
defined by
\begin{equation*}
  (A^*f^*_b)(e_a) = (-1)^{\pa(A)q_b} f^*_b(Ae_a) .
\end{equation*}
The supertranspose has the following properties:
\begin{enumerate}[1)]
\item $(A+B)^*=A^*+B^*$;
\item $(AB)^*=(-1)^{\pa(A)\pa(B)} B^* A^*$;
\item $(A^*)^*=(-1)^{\pa(A)} A$.
\end{enumerate}
Although $A\mapsto A^*$ is not in general an involution, its square
is.

If $V$ is a superspace, $\Pi V$ is the superspace obtained by
exchanging the even and odd subspaces of $V$. Denote $\Pi V^*$ by
$V^\circ$. If $A:V\to W$ is a morphism of superspaces, let $A^\Pi$ be
the induced morphism from $\Pi V$ to $\Pi W$. Let
$A^\circ=A^{\Pi*}:W^\circ\to V^\circ$. This operation has the
following properties:
\begin{enumerate}[1)]
\item $(A+B)^\circ=A^\circ+B^\circ$;
\item $(AB)^\circ=(-1)^{\pa(A)\pa(B)} B^\circ A^\circ$;
\item $(A^\circ)^\circ=A$.
\end{enumerate}

If $V$ is a finite-dimensional superspace, then so is its space of
endomorphisms $\End(V)$. The Berezinian is a rational function on
$\End(V)$ that specializes to the determinant when $V$ is concentrated
in even degree, and to the inverse of the determinant when $V$ is
concentrated in odd degree. To give a formula for the Berezinian, we
break $A$ up into blocks mapping between the parity homogeneous
components of $V$: $A_{pq}$, $p,q\in\{0,1\}$, maps $V_q$ to $V_p$. We
have
\begin{equation*}
  \Ber(A) = \frac{\det(A_{00})}{\det(A_{11}-A_{10}A_{00}^{-1}A_{01})}
  \frac{\det(A_{00}-A_{01}A_{11}^{-1}A_{10})}{\det(A_{11})} .
\end{equation*}
The Berezinian has the following properties:
\begin{enumerate}[1)]
\item $\Ber(\Id)=1$;
\item $\Ber(AB)=\Ber(A)\Ber(B)$;
\item $\Ber(A^*)=\Ber(A)$ and $\Ber(A^\circ)=\Ber(A)^{-1}$.
\end{enumerate}
The general linear supergroup $\GL(V)$ is the Zariski open subset of
$\End(V)$ where $\Ber(A)\ne\{0,\infty\}$.

Let $V$ be a superspace with an odd symplectic form, that is, a
bilinear pairing $\om(-,-):V\times V\to\C$ satisfying
\begin{enumerate}[a)]
\item $\om(v,w)=0$ unless $\pa(v)+\pa(w)=1$;
\item $\om(v,w)=-\om(w,v)$.
\end{enumerate}
A polarization of $V$ is a decomposition of $V$ into a direct sum
\begin{equation*}
  V = L \oplus M ,
\end{equation*}
where $L$ and $M$ are Lagrangian subspaces for the form $\om$, that
is, maximal isotropic subspaces. The symplectic form induces a natural
isomorphism $M\cong L^\circ$. We will work in a Darboux basis of $V$,
consisting of a homogeneous basis $\{e_a\}_{a\in I}$ of $L$ and the
dual basis $\{f_a\}_{a\in I}$ of $L^\circ$, where $\pa(f_a)=p_a+1$.

The odd symplectic supergroup $\PiSP(V,\om)\subset\GL(V)$ consists of
transformations preserving the symplectic form $\om$. Decomposing an
endomorphism $A:V\to V$ into blocks
\begin{equation*}
  A =
  \begin{pmatrix}
    P & Q \\ R & S
  \end{pmatrix}
\end{equation*}
where $P:L\to L$, $Q:L^\circ\to L$, $R:L\to L^\circ$ and
$S:L^\circ\to L^\circ$, we may express the conditions to preserve the
symplectic form as the following matrix equations:
\begin{align*}
  S^\circ P &= Q^\circ R + \Id_L , & P^\circ R &= R^\circ P , \\
  Q^\circ S &= S^\circ Q , & P^\circ S &= R^\circ Q + \Id_{L^\circ} .
\end{align*}
These equations define a superquadric $\QQ(V,\om)$ in $\End(V)$,
closed under composition. The supergroup $\PiSP(V,\om)$ is the open
subset of $\QQ(V,\om)$ where $\Ber(A)$ is nonzero.

\begin{proposition}[Khudaverdian and Voronov {\cite{KV}*{Theorem 4}}]
  The restriction of the Berezinian to $\QQ(V,\om)$ is the rational
  function $\Ber(P)^2$.
\end{proposition}
\begin{proof}
  From the equation
  \begin{equation*}
    \begin{pmatrix}
      P & Q \\ R & S
    \end{pmatrix}
    =
    \begin{pmatrix}
      \Id_L & 0 \\ RP^{-1} & \Id_{L^\circ}
    \end{pmatrix}
    \begin{pmatrix}
      P & Q \\ 0 & S - RP^{-1}Q
    \end{pmatrix}
  \end{equation*}
  we see that
  \begin{equation*}
    \Ber(A) = \Ber(P)\Ber(S-RP^{-1}Q) .
  \end{equation*}
  By the equation
  \begin{align*}
    P^\circ(S-RP^{-1}Q) &= P^\circ S - (P^\circ R)P^{-1}Q \\
                       &= \bigl( R^\circ Q + \Id_{L^\circ} \bigr) -
                         (R^\circ P)P^{-1}Q \\
                        &= \Id_{L^\circ} ,
  \end{align*}
  we see that $\Ber(S-RP^{-1}Q)=\Ber(P)$, and the result follows.
\end{proof}

Unlike for a superspace with an even orthosymplectic form, the
Lagrangian Grassmannian of a superspace with an odd symplectic form is
not connected. The above result holds on each component of this
Grassmannian. On the component containing the even parity subspace
$V_0$ of $V$, the proposition shows that $\Ber(A)=\det(P)^2$, and we
may conclude that the restriction of the Berezinian to the quadric
$\QQ(V,\om)$ is a polynomial in the entries of $A$.

Let $\Ber^{1/2}(A)$ be the function $\Ber(P)$ on $\QQ(V,\om)$.
\begin{corollary}
  The function $\Ber^{1/2}(A)$ satisfies
  \begin{equation*}
    \Ber^{1/2}(AB)=\Ber^{1/2}(A)\Ber^{1/2}(B) .
  \end{equation*}
\end{corollary}
\begin{proof}
  This formula holds up to a rational function in the entries of $A$
  and $B$ whose square equals $1$, and is thus constant. Taking $A$
  and $B$ equal to $\Id$, the result follows.
\end{proof}

In this article, all superspaces are $\Z$-graded: in physics, this
grading is referred to as the ghost number. All of the odd symplectic
forms considered will have ghost number $-1$: that is, $\om(v,w)$
vanishes on a pair of vectors $v$ and $w$ homogeneous with respect to
the ghost number unless $\gh(v)+\gh(w)=1$.

\section{Odd symplectic supermanifolds}

We now recall the definition of the line bundle of half-forms
$\Om^{1/2}$ on a supermanifold $M$ with odd symplectic form $\om$, and
the differential operator $\Delta$ acting on sections of
$\Om^{1/2}$. This material is taken from Khudaverdian
\cite{Khudaverdian}.

We start with a supermanifold $M$, modelled on a superspace $V$. The
line bundle $\Om$ over $M$ associated to the character $\Ber(A)^{-1}$
of $\GL(V)$ is called the bundle of integral forms. An odd
non-degenerate 2-form $\om$ on $M$ induces a $\PiSP(V,\om)$-structure
on the tangent superbundle $TM$. The line bundle $\Om^{1/2}$ over $M$
associated to the character $\Ber^{1/2}(A)^{-1}$ of $\PiSP(V,\om)$ is
called the bundle of half-forms. (Khudaverdian \cite{Khudaverdian}
actually studies a similar bundle, associated to the character
$|\Ber^{1/2}(A)|^{-1}$, called the bundle of half-densities, but his
results hold if this bundle is replaced by the bundle of half-forms.)

A Darboux coordinate chart on $M$ is a coordinate chart
$\{x^a,\xi_a\}_{a\in I}$ such that the tangent vectors
$\{\p/\p x^a\}_{a\in I}$ and $\{\p/\p\xi_a\}_{a\in I}$ span Lagrangian
subspaces of the tangent superbundle $TM$, and
\begin{equation*}
  \om(\p/\p\xi_a,\p/\p x^b) = \delta_a^b ,
\end{equation*}
or equivalently,
\begin{equation*}
  \om = \sum_{a\in I} (-1)^{p_a} \, d\xi_a \wedge d x^a .
\end{equation*}
Such coordinate charts exist locally if the 2-form $\om$ is closed, in
which case, $M$ is called an odd symplectic supermanifold. Let
$\{x^a,\xi_a\}_{a\in I}$ be a Darboux coordinate chart: a total order
of the indices $I$ (or equivalently, a bijection between $I$ and the
set of natural numbers $\{1,\dots,n\}$, where $n$ is the cardinality
of $I$) yields a nonvanishing section
\begin{equation*}
  dx = dx^1 \dots dx^n
\end{equation*}
of the line bundle $\Om^{1/2}$ of half-forms. Let $\Delta_0$ be the
second-order differential operator
\begin{equation*}
  \Delta_0 = \sum_{a\in I} (-1)^{p_a} \frac{\p^2f}{\p x^a\p\xi_a} .
\end{equation*}
Define a second-order differential operator $\Delta$ on sections of
$\Om^{1/2}$ as follows: a section $\sigma$ of $\Om^{1/2}$ may be
written in the Darboux chart as $f\,dx$, and we define
\begin{equation*}
  \Delta\sigma = \Delta_0f \, dx .
\end{equation*}
It is clear that this does not depend on the choice of total ordering
of $I$, since a change of ordering only changes $dx$, and hence $f$,
by a sign.

Given a function $f$ on $M$, denote by $m(f)$ the operation of
multiplication by $f$ on the sheaf of sections of $\Om^{1/2}$. The
following properties of the operator $\Delta$ are easily derived from
the explicit formula in a Darboux coordinate chart:
\begin{enumerate}[1)]
\item $\Delta^2=\half[\Delta,\Delta]=0$;
\item $\Delta$ is a second-order differential operator, that is, if
  $f$, $g$ and $h$ are functions on $M$, then
  \begin{equation*}
    [[[\Delta,m(f)],m(g)],m(h)] = 0 .
  \end{equation*}
\end{enumerate}

The Batalin--Vilkovisky antibracket is the Poisson bracket associated
to the odd symplectic form $\om$, given by the formula
\begin{equation*}
  m \bigl( (f,g) \bigr) = (-1)^{\pa(f)} [[\Delta,m(f)],m(g)] .
\end{equation*}
\begin{proposition}
  The antibracket is antisymmetric
  \begin{equation*}
    (g,f) = - (-1)^{(\pa(f)+1)(\pa(g)+1)} \, (f,g)
  \end{equation*}
  and satisfies the Jacobi relation
  \begin{equation*}
    (f,(g,h)) = ((f,g),h) + (-1)^{(\pa(f)+1)\pa(g)} \, (g,(f,h)) .
  \end{equation*}
\end{proposition}
\begin{proof}
  To prove antisymmetry, we use the formula $[f,g]=0$:
  \begin{multline*}
    m \bigl( (f,g) + (-1)^{(\pa(f)+1)(\pa(g)+1)} \, (f,g) \bigr) \\
    \begin{aligned}
      &= (-1)^{\pa(f)} [[\Delta,m(f)],m(g)] -
      (-1)^{\pa(f)\pa(g)+\pa(f)} [[\Delta,m(g)],m(f)] \\
      &= (-1)^{\pa(f)} [[\Delta,m(f)],m(g)] + [m(f),[\Delta,m(g)]] \\
      &= (-1)^{\pa(f)} [\Delta,[m(f),m(g)]] = 0 .
    \end{aligned}
  \end{multline*}
  To prove the Jacobi relation, we use the formulas
  $[\Delta,[\Delta,m(f)]]=\half[[\Delta,\Delta],m(f)]=0$ and
  $[[[\Delta,m(f)],m(g)],m(h)]=0$:
  \begin{align*}
    0 &= - [\Delta,[[[\Delta,m(f)],m(g)],m(h)]] \\
      &= (-1)^{\pa(f)} [[[\Delta,m(f)],[\Delta,m(g)]],m(h)] \\
      &+ (-1)^{\pa(f)+\pa(g)} [[[\Delta,m(f)],m(g)],[\Delta,m(h)]] \\
      &= (-1)^{\pa(f)} [[\Delta,m(f)],[[\Delta,m(g)],m(h)]] \\
      &+ (-1)^{(\pa(f)+1)\pa(g)} [[\Delta,m(g)],[[\Delta,m(f)],m(h)]] \\
      & + (-1)^{\pa(h)(\pa(f)+\pa(g)+1)}
        [[\Delta,m(h)],[[\Delta,m(f)],m(g)]] \\
      &= (-1)^{\pa(g)} (f,(g,h)) +
        (-1)^{\pa(f)\pa(g)+\pa(f)} (g,(f,h)) \\
      & + (-1)^{\pa(h)(\pa(f)+\pa(g))+\pa(f)} (h,(f,g)) \\
      &= (-1)^{\pa(g)} \bigl( (f,(g,h)) -
        (-1)^{(\pa(f)+1)(\pa(g)+1)} (g,(f,h)) - ((f,g),h) \bigr) .
        \qedhere
  \end{align*}
\end{proof}

Let $\H_f$ be the first-order differential operator given by the
formula
\begin{equation*}
  \H_f = (-1)^{\pa(f)} [\Delta,m(f)] .
\end{equation*}

\begin{proposition}
  \mbox{} \\[-15pt]
  \begin{enumerate}[1),font=\upshape]
  \item $[\H_f,m(g)]=m\bigl( (f,g) \bigr)$
  \item $\H_{fg}=m(f) \H_g + (-1)^{\pa(f)\pa(g)} m(g) \H_f +
    (-1)^{\pa(g)} m\bigl( (f,g) \bigr)$
  \item $\H_{(f,g)}=[\H_f,\H_g]$
  \end{enumerate}
\end{proposition}
\begin{proof}
  The first formula follows immediately from the definition of $\H_f$.
  The second formula is proved as follows:
  \begin{align*}
    \H_{fg}
    &= (-1)^{\pa(f)+\pa(g)} [\Delta,m(fg)] \\
    &= (-1)^{\pa(g)} m(f)[\Delta,m(g)] + (-1)^{\pa(f)+\pa(g)} [\Delta,m(f)]m(g) \\
    &= m(f)\H_g + (-1)^{\pa(g)} \H_f m(g) \\
    &= m(f) \H_g + (-1)^{\pa(g)+(\pa(f)+1)\pa(g)} m(g) \H_f + (-1)^{\pa(g)}
      m\bigl( (f,g) \bigr) .
      \intertext{To prove the third formula, we argue as follows:}
      [\H_f,\H_g] &= (-1)^{\pa(f)+\pa(g)} [[\Delta,m(f)],[\Delta,m(g)]] \\
    &= (-1)^{\pa(g)+1} [\Delta,[[\Delta,m(f)],m(g)]]
      - (-1)^{\pa(g)+1} [[\Delta,[\Delta,m(f)]],m(g)] .
  \end{align*}
  The first term equals $\H_{(f,g)}$, and the second term vanishes.
\end{proof}

If $f$ is a function of odd parity on $M$, the operator
$g\mapsto(f,g)$ is an even vector field $H_f$ on $M$, called the
Hamiltonian vector field associated to $f$. Under the infinitesimal
flow $\exp(\epsilon H_f)$, the Darboux coordinate chart transforms to
\begin{align*}
  x^a + \epsilon (f,x^a)
  &= x^a - \epsilon \frac{\p f}{\p\xi_a} &
  \xi_a + \epsilon (f,\xi_a)
  &= \xi_a + \epsilon \frac{\p f}{\p x^a} .
\end{align*}
The section $dx$ of the bundle of half-forms transforms to
\begin{equation*}
  \Ber\Bigl( \Id - \epsilon \frac{\p^2f}{\p x^j\p\xi_a} \Bigr) \, dx =
  dx - \epsilon \Delta (f\,dx) .
\end{equation*}
 A half-form $\sigma=g\,dx$ transforms to
\begin{align*}
  \bigl( g + \epsilon (f,g) \bigr) \, dx - \epsilon g \Delta(f\,dx)
  &= \sigma - \epsilon [[\Delta,f],g]\,dx - \epsilon g \Delta(f\,dx) \\
  &= \sigma + \epsilon \H_f\sigma .
\end{align*}
We may interpret the differential operator $\H_f$ as the lift of the
vector field $H_f$ to the bundle of half-forms. The operator $\Delta$
is invariant under this action, by the formula
\begin{equation*}
  [\H_f,\Delta] = - [\Delta,[\Delta,m(f)]] = 0 .
\end{equation*}

A diffeomorphism preserving the antibracket is called a
(Batalin--Vilkovisky) canonical transformation. These transformations
form a pseudogroup, which is generated by nonautonomous Hamiltonian
flows (associated to time-dependent Hamiltonians), in the sense that
for any canonical transformation $f:U\to V$, each point $x\in U$ has
an open neighborhood on which the restriction of $f$ may be written as
the composition of flows associated to nonautonomous Hamiltonian
vector fields. It follows that the operator $\Delta$ is independent of
the choice of Darboux coordinate chart. This is the strategy adopted
by Khudaverdian in his proof of this theorem
\cite{Khudaverdian}*{Section~2}; \v{S}evera \cite{Severa} has given
another proof, which identifies Khudaverdian's operator $\Delta$ with
the differential on the $E_2$ page of the spectral sequence associated
to the Hodge filtration (filtration by degree of differential forms)
in the de Rham complex of $M$ with deformed differential $d+\om$.

An orientation of an odd symplectic manifold $M$ is a
nowhere-vanishing section $\sigma$ of the bundle of half-forms
$\Om^{1/2}$ such that $\Delta\sigma=0$ (Behrend and Fantechi
\cite{BF}). In particular, $\sigma$ defines a global trivialization of
$\Om^{1/2}$. If we choose a Darboux coordinate chart and express
$\sigma$ as $e^Sdx$, we may write the equation $\Delta\sigma=0$ as
\begin{equation*}
  \Delta_0S + \half (S,S) = 0 .
\end{equation*}
This equation is known as the quantum master equation. In applications
of this equation to quantum field theory, there is an additional
parameter $\hbar$. In this case, $\sigma$ equals $e^{S/\hbar}dx$,
where $S$ is itself a power series in $\hbar$:
\begin{equation*}
  S = \sum_{n=0}^\infty \hbar^n S_n .
\end{equation*}
In this setting, the quantum master equation becomes
\begin{equation*}
  \hbar \Delta_0S + \half (S,S) = 0 .
\end{equation*}
Expanding in powers of $\hbar$, this equation is seen to be equivalent
to the series of equations
\begin{equation*}
  \begin{cases}
    (S_0,S_0)=0 , & \\
    \displaystyle
    \Delta_0 S_{n-1} + (S_0,S_n) + \frac12 \sum_{i=1}^{n-1}
    (S_i,S_{n-i}) = 0 , & n>0 .    
  \end{cases}
\end{equation*}

\section{Simplicial supermanifolds}

A simplicial supermanifold $M_\bull$ consists of the following data:
for each $k\ge0$, $M_k$ is a supermanifold, and there are face maps
$d_i:M_k\to M_{k-1}$ and degeneracy maps $s_i:M_k\to M_{k+1}$
satisfying the usual simplicial relations. Let $[k]=\{0,\dots,k\}$ be
the set of vertices of the $k$-simplex. If $\mu:[k]\to[\ell]$ is a
function preserving the ordering of the vertices, then there is a
differentiable map $\mu^*:M_\ell\to M_k$ satisfying
$(\mu\nu)^*=\nu^*\mu^*$. The face map $d_i$ is associated to the
function
\begin{align*}
  \mu(j)
  &=
  \begin{cases}
    j , & j<i , \\
    j+1 , & j\ge i ,
  \end{cases}
            \intertext{while the degeneracy map $s_i$ is associated
            to the function}
  \mu(j)
  &=
  \begin{cases}
    j , & j\le i , \\
    j-1 , & j>i .
  \end{cases}
\end{align*}            

For example, suppose that $M$ is a supermanifold and
$\CU=\{U_\alpha\}$ is a locally finite open cover of $M$. If
$(\alpha_0\dots\alpha_k)$ is a sequence of indices of the open sets in
the cover $\CU$ of $M$, we denote by
\begin{equation*}
  U_{\alpha_0\dots\alpha_k}
\end{equation*}
their intersection. We obtain a simplicial supermanifold $\CU_\bull$
by setting
\begin{equation*}
  \CU_k = \coprod_{\alpha_0\dots\alpha_k} U_{\alpha_0\dots\alpha_k}
  ,
\end{equation*}
where $\mu^*$ is the inclusion of
$U_{\alpha_0\dots\alpha_\ell}\subset\CU_\ell$ into
$U_{\alpha_{\mu(0)}\dots\alpha_{\mu(k)}}\subset\CU_k$. In particular, the
face map $d_i:\CU_k\to\CU_{k-1}$ is the open embedding of
$U_{\alpha_0\dots\alpha_k}$ into
\begin{equation*}
  U_{\alpha_0\dots\widehat{\alpha}_i\dots\alpha_k} \subset \CU_{k-1} ,
\end{equation*}
while the degeneracy map $s_i:\CU_k\to\CU_{k+1}$ is the identification
of $U_{\alpha_0\dots\alpha_k}$ with
\begin{equation*}
  U_{\alpha_0\dots\alpha_i\alpha_i\dots\alpha_k} \subset \CU_{k+1} .
\end{equation*}
Every map $\mu^*$ may be factored into a finite sequence of face maps
followed by a sequence of degeneracy maps, so it actually suffices to
consider just these two sets of maps.

The $k$-simplex $\Delta^k$ is the convex hull of the unit coordinate
vectors in $\R^{[k]}=\R^{\{0,\dots,k\}}$. We denote the coordinates on
$\R^{[k]}$ by $(t_0,\dots,t_k)$; on $\Delta^k$, they satisfy
$t_0+\dots+t_k=1$ and $t_i\ge0$. A function $\mu:[k]\to[\ell]$
preserving the order of the vertices induces a map
$\mu_*:\Delta^k\to\Delta^\ell$: the vertices of $\Delta^k$ are mapped
to the vertices of $\Delta^\ell$ following the function $\mu$, and the
map is the affine extension to the convex hull of these points. In
particular, the coface map $d^i:\Delta^{k-1}\to\Delta^k$,
$0\le i\le k$, is given by formula
\begin{equation*}
  d^i(t_0,\dots,t_{k-1}) = (t_0,\dots,t_{i-1},0,t_i,\dots,t_{k-1}) ,
\end{equation*}
and the codegeneracy map $s^i:\Delta^{k+1}\to\Delta^k$, $0\le i\le k$,
is given by formula
\begin{equation*}
  s^i(t_0,\dots,t_{k+1}) = (t_0,\dots,\widehat{t}{}_i,\dots,t_{k+1}) .
\end{equation*}
These maps go in the opposite direction to the maps in a simplicial
manifold: $\Delta^\bull$ is an example of a cosimplicial manifold
(with corners). Let $\Om_k$ be the de Rham complex of $\Delta^k$,
with differential $\delta$: this is a simplicial complex of
superspaces (and even a simplicial differential graded algebra).

If $V$ is a vector bundle on a supermanifold $M$, there is an
inclusion with dense image
\begin{equation*}
  \Gamma(M,V) \o \Om_k \to
  \Gamma(M\times\Delta^k,V\boxtimes\Lambda^*T^*\!\Delta^k) .
\end{equation*}
The target of this morphism may be thought of as a completed tensor
product
\begin{equation*}
  \Gamma(M\times\Delta^k,V\boxtimes\Lambda^*T^*\!\Delta^k) =
  \Gamma(M,V) \ho \Om_k .
\end{equation*}

\begin{definition}
  A morphism $f:M\to N$ of supermanifolds is \textbf{\'etale} if it is
  locally an open embedding, or equivalently, the tangent maps
  $d_x\mu^*$ are isomorphisms at all points $x\in M$, and a
  \textbf{cover} if it is \'etale and surjective.
\end{definition}

A simplicial odd symplectic supermanifold is a simplicial manifold
$M_\bull$ such that
\begin{enumerate}[a)]
\item each supermanifold $M_k$ is odd symplectic, and
\item the morphisms $\mu^*:M_\ell\to M_k$ are \'etale and preserve the
  odd symplectic structures.
\end{enumerate}
An example of a simplicial odd symplectic supermanifold is the
\v{C}ech nerve $\CU_\bull$ associated to an open cover $\CU$ of an odd
symplectic supermanifold.

Associated to a simplicial odd symplectic supermanifold are the
cosimplicial commutative superalgebra $\CO(\M_\bull)$ of functions on
$M_\bull$ and the cosimplicial $\CO(\M_\bull)$-supermodule
$\Om^{1/2}(\M_\bull)$. The Thom-Whitney normalization of
$\CO(\M_\bull)$ is the differential graded commutative superalgebra
\begin{multline*}
  \Tot \CO(M_\bull) = \Bigl\{ \bigl( f_k \bigr) \in \prod_{k=0}^\infty
  \CO(M_k) \ho \Om_k \Bigm| \\
  \text{for all $\mu:[k]\to[\ell]$, we have $(\mu^*\ho1)f_\ell =
    (1\ho\mu_*)f_k \in \CO(M_k) \ho \Om_\ell$} \Bigr\}
\end{multline*}
with differential $\delta$ induced by the de~Rham differential
$\delta$ on $\Om_k$ for different $k$. In \cite{covariant}, we showed
that topological terms may be incorporated into the
Batalin--Vilkovisky formalism by replacing the classical master
equation $\half(S,S)=0$ for an element of $\CO(M)$ by the classical
master equation
\begin{equation*}
  \delta S + \half(S,S) = 0
\end{equation*}
in $\Tot \CO(M_\bull)$. (In \cite{covariant}, we denote this
totalization by $\Tot_\TW$; since it is the only totalization functor
employed in this article, we will write $\Tot$ instead.)

The Thom-Whitney normalization of the cosimplicial superspace
$\Om^{1/2}(M_\bull)$ is the differential graded
$\Tot \CO(M_\bull)$-module
\begin{multline*}
  \Tot \Om^{1/2}(M_\bull) = \Bigl\{ \bigl( \sigma_k \bigr) \in
  \prod_{k=0}^\infty \Om^{1/2}(M_k) \ho \Om_k \Bigm| \\
  \text{for all $\mu:[k]\to[\ell]$, we have
    $(\mu^*\ho1)\sigma_\ell = (1\ho\mu_*)\sigma_k \in
    \Om^{1/2}(M_k) \ho \Om_\ell$} \Bigr\} .
\end{multline*}
The operator $\delta+\Delta$ descends to $\Tot \Om^{1/2}(M_\bull)$,
turning it into a complex. In this setting, the quantum master
equation becomes
\begin{equation*}
  ( \delta + \Delta ) \sigma = 0 .
\end{equation*}
In the presence of the parameter $\hbar$, this is modified to
$(\delta+\hbar\Delta)\sigma=0$.

\section{Families of Lagrangian submanifolds}

Suppose that $M$ is an odd symplectic supermanifold, and that
$\iota:L\subset M$ is a Lagrangian submanifold. In other words, the
restriction of the odd symplectic form $\om$ to $L$ induces an
isomorphism between the tangent superbundle $TL$ of $L$ and the
conormal superbundle $N^*\!L$. In particular, a Lagrangian submanifold
is coisotropic: the ideal of functions vanishing on $L$ is closed
under the antibracket.

It is a basic result of odd symplectic geometry that given a
Lagrangian submanifold $L$ and a point $x\in L$, there is a Darboux
coordinate chart $U$ around $x$ such that $U\cap L$ is the submanifold
$\{\xi_a=0\}$. In other words, a neighborhood of $x$ is identified
with a neighbourhood of $x$ in the odd cotangent bundle $\Pi
T^*\!L$. The proof is identical to the proof in the even case
(Weinstein \cite{Weinstein}).

The restriction $\iota^*\Om^{1/2}$ of the bundle of half-forms
$\Om^{1/2}$ on $M$ to a Lagrangian submanifold $L$ is isomorphic to
the bundle of integral forms on $L$, that is, the Berezinian bundle
$\Ber(T^*\!L)$ of the cotangent bundle of $L$. (If $L$ is a manifold,
this is the same as the bundle of differential forms of top degree on
$L$.) The integral is an invariantly defined linear form on the space
of integral forms of compact support on $L$: thus, the restriction map
induces a linear form $\tintL : \Om^{1/2}_c(M) \to \C$ on the bundle
of compactly supported half-forms on $M$.

We now consider the generalization of this construction when
$\iota:L\times\Delta^k\to M$ is a family of Lagrangian submanifolds
parametrized by the $k$-simplex $\Delta^k$. Taking the derivative of
the map $\iota$ in the simplicial direction, we obtain a family of
vector fields $X\in\Gamma(L,\iota^*TM)\ho \Om_k$ over $L$,
parametrized by one-forms on $\Delta^k$. Take the contraction with the
odd symplectic form $\om$
\begin{equation*}
  X\inprod\om \in \Gamma(L,\iota^*T^*\!M) \ho \Om_k
\end{equation*}
to convert this vector field into a differential in the ambient
manifold $M$. Applying the bundle map
$\iota^*:\iota^*T^*\!M\to T^*\!L$ adjoint to the differential
$\iota_*:TL\to \iota^*TM$, we obtain a family of 1-forms on $L$:
\begin{equation*}
  \iota^*(X\inprod\om) \in \Om^1(L) \ho \Om_k .
\end{equation*}
The condition that $L_t$ is Lagrangian for all $t\in\Delta^k$ is
equivalent to the condition that this family of one-forms is closed:
\begin{equation*}
  d\iota^*(X\inprod\om) = 0 \in \Om^2(L) \ho \Om_k .
\end{equation*}
Here, we denote by $d$ the differential in the first factor $L$ of a
product $L\times\Delta^k$ of a supermanifold with a simplex, and by
$\delta$ the differential in the second factor $\Delta^k$. Since the
de Rham cohomology of $L$ vanishes in nonzero ghost number, there is a
uniquely determined family of one-forms
\begin{equation*}
  \eta \in \CO(L) \ho \Om_k
\end{equation*}
such that $d\eta=\iota^*(X\inprod\om)$ and $\delta\eta=0$.

Let us rewrite this equation in a Darboux coordinate system on
$M$. Thus, suppose that $L$ has coordinates $\x^a$ and $\iota$ is
given in a Darboux coordinate system on $M$ by the equations
\begin{equation*}
  \begin{cases}
    x^a = x^a(\x,t) , & \\
    \xi_a = \xi_a(\x,t) , &
  \end{cases}
\end{equation*}
where $x^a(\x,0)=\x^a$ and $\xi_a(\x,0)=0$. The one-form
$\eta=\eta_i(x,t)\, dt^i$ satisfies the differential equation
\begin{equation*}
  \frac{\p\xi_a(\x,t)}{\p t^i} = \frac{\p\eta_i(\x,t)}{\p\x^a} .
\end{equation*}
This implies that $\delta\eta$ is independent of $\x$: thus, if
$\delta\eta$ vanishes at any point in $L$, it vanishes
everywhere. This may always be arranged, by replacing $\eta$ by
\begin{equation*}
  \eta + \sum_{i=0}^k t^i d\eta_i(\x_0,t) .
\end{equation*}

Since we are only concerned with families of Lagrangian submanifolds
up to reparametrization, we may assume that, at least in a
neighborhood of $(x_0,0)\in L\times\Delta^k$, we have $x^a=\x^a$. In
this case, $L$ is (locally) a family of sections of the odd cotangent
bundle $\Pi T^*\!L$:
\begin{equation*}
  \begin{cases}
    x^a = \x^a , & \\
    \xi_a = \xi_a(\x,t) . &
  \end{cases}
\end{equation*}

The following theorem is a mild generalization of a result of
Mikhalkov and Schwarz \cite{MS}*{(3.6)}.
\begin{theorem}
  \label{MS}
  Let $\iota:L\times\Delta^k\to M$ be a proper family of Lagrangian
  submanifolds of $M$, and let $\sigma\in\Om^{1/2}_c(M) \ho \Om_k$ be
  a family of compactly supported half-forms on $M$. Then
  \begin{equation*}
    \delta \int_L e^{-\eta/\hbar} \, \iota^*\sigma =
    \int_L e^{-\eta/\hbar} \, \iota^*(\delta+\hbar\Delta)\sigma .
  \end{equation*}
\end{theorem}
\begin{proof}
  Let $x_0$ be a point in $L$, and consider a Darboux coordinate chart
  $(x^a,\xi_a)$ around $\iota(x_0,0)\in M$. If $\sigma = f \, dx$, we
  have
  \begin{equation*}
    (\iota^*\sigma)(\x,t) = f(x^a(\x,t),\xi_a(\x,t),t,dt) \,
    \Ber\Biggl( \frac{\p x^a(\x,t)}{\p\x^b} \Biggr) \, d\x .
  \end{equation*}
  We may assume that $x^a(\x,t)=\x^a$, in which case we have
  \begin{equation*}
    (\iota^*\sigma)(\x,t) = f(\x^a,\xi_a(\x,t),t,dt) \, d\x .
  \end{equation*}
  Applying the differential $\delta$ and multiplying by the
  inhomogeneous differential form $e^{-\eta/\hbar}$, we obtain
  \begin{multline*}
    \delta (e^{-\eta/\hbar} \, \iota^*\sigma)(\x,t) =
    e^{-\eta(\x,t)/\hbar} \left( \iota^*\delta\sigma + \sum_{a\in I}
      \delta\xi_a(\x,t) \frac{\p
        f(\x^a,\xi_a(\x,t),t,dt)}{\p\xi_a} \, d\x \right) \\
    \begin{aligned}
      &= e^{-\eta(\x,t)/\hbar} \left( \iota^*\delta\sigma + \sum_{a\in
          I} \frac{\p\eta(\x,t)}{\p\x^a} \frac{\p
          f(\x^a,\xi_a(\x,t))}{\p\xi_a} \, d\x \right) \\
      &= e^{-\eta(\x,t)/\hbar} \, \iota^*(\delta+\hbar\Delta)\sigma -
      \hbar \sum_{a\in I} (-1)^{p_a} \frac{\p}{\p\x^a} \left(
        e^{-\eta(\x,t)/\hbar} \, \frac{\p
          f(\x^a,\xi_a(\x,t))}{\p\xi_a} \right) d\x .
    \end{aligned}
  \end{multline*}
  Integrating over $L$, we obtain the result.
\end{proof}

\section{Lagrangian submanifolds of simplicial odd symplectic
  supermanifolds}

\label{family}

In this section, we explain what we mean by a Lagrangian submanifold
of a simplicial odd symplectic supermanifold $M_\bull$: this is our
formulation in the Batalin--Vilkovisky setting of a global gauge
condition pieced together from local gauge conditions.

In order to define a Lagrangian submanifold in this generalized sense,
we start with a family of graded supermanifolds $\{L_k\}_{k\ge0}$
indexed by the natural numbers. For each function $\mu:[k]\to[\ell]$
preserving the order of the vertices, we are given a morphism of
graded supermanifolds
\begin{equation}
  \label{flexible}
  \begin{tikzcd}
    L_\ell \times \Delta^k \ar[dr] \ar[rr,"\mu^*"] & & L_k \times
    \Delta^k \ar[dl] \\
    & \Delta^k &
  \end{tikzcd}
\end{equation}
such that if $\mu:[k]\to[\ell]$ and $\nu:[j]\to[k]$ are a pair of
functions preserving the order of the vertices, the following diagram
commutes:
\begin{equation*}
  \begin{tikzcd}
    L_\ell \times \Delta^j \ar[dr,"(\mu\nu)^*"']
    \ar[rr,"\mu^*\times_{\Delta^k}\Delta^j"]
    & & L_k \times \Delta^j \ar[dl,"\nu^*"] \\
    & L_j \times \Delta^j &
  \end{tikzcd}
\end{equation*}
The fibred product $\mu^*\times_{\Delta^k}\Delta^j$ is taken with
respect to the morphism of simplices $\nu_*:\Delta^j\to\Delta^k$.

The other data needed to specify a Lagrangian submanifold of $M_\bull$
are morphisms $\iota_k:L_k\times\Delta^k\to M_k$ satisfying the
following conditions:
\begin{enumerate}[a)]
\item for every point $\t\in\Delta^k$, the restriction of $\iota_k$ to
  $L_k\times\t$ is a proper Lagrangian embedding;
\item for each morphism $f:[k]\to[\ell]$ in the simplicial category,
  the following diagram commutes:
  \begin{equation*}
    \begin{tikzcd}[column sep=huge]
      L_\ell\times\Delta^k \ar[d,"\mu^*"']
      \ar[r,"L_\ell\times\mu_*"] & L_\ell\times\Delta^\ell
      \ar[r,"\iota_\ell"] & M_\ell \ar[d,"\mu^*"] \\
      L_k\times\Delta^k \ar[rr,"\iota_k"'] & & M_k
    \end{tikzcd}
  \end{equation*}
\end{enumerate}
We only consider the case in which $M_\bull=\CU_\bull$ is the \v{C}ech
nerve associated to an open cover $\CU=\{U_\alpha\}$ of an odd
symplectic supermanifold $M$. In this case, $L_k$ decomposes into a
disjoint union
\begin{equation*}
  L_k = \coprod_{\alpha_0\dots\alpha_k} L_{\alpha_0\dots\alpha_k} ,
\end{equation*}
and $\iota_k$ decomposes into families of proper Lagrangian embeddings
\begin{equation*}
  \iota_{\alpha_0\dots\alpha_k} : L_{\alpha_0\dots\alpha_k} \times
  \Delta^k \to \CU_{\alpha_0\dots\alpha_k} .
\end{equation*}
Denote by
$\eta_{\alpha_0\dots\alpha_k} \in
\CO(L_{\alpha_0\dots\alpha_k})\ho\Om_k \subset
\Om^1(L_{\alpha_0\dots\alpha_k}\times\Delta^k)$ the one-form
associated to the family of Lagrangian submanifolds determined by
$\iota_k$.

The main result of this article is the definition of a linear form
\begin{equation*}
  \T : \Tot\Om^{1/2}(\CU_\bull) \to \C
\end{equation*}
of degree $0$, which is closed in the sense that
\begin{equation*}
  \T\bigl( (\delta+\hbar\Delta)\sigma \bigr) = 0 .
\end{equation*}
The formula for $\T$ specializes, in the case of a Lagrangian
submanifold $L\subset M$, to the formula of Batalin and Vilkovisky,
\begin{equation*}
  \T(\sigma) = \int_L \iota^*\sigma .
\end{equation*}
They interpret $\T(\sigma)$ as the partition function for a quantum
field theory, and $L$ as a gauge condition. Our extension of their
formula allows the use of more general gauge conditions, and opens the
door to the use of the Batalin--Vilkovisky formalism when the action
contains topological terms.

The definition of $\T$ depends on the auxilliary data of a partition of
unity $\{\varphi_\alpha\}$ for the cover $\CU$. In other words,
$\varphi_\alpha\in\CO_c(U_\alpha)$, and
\begin{equation}
  \label{partition}
  \sum_\alpha \varphi_\alpha = 1 .
\end{equation}
Denote the commutator $[\Delta,m(\varphi_\alpha)]=\H_{\varphi_\alpha}$
by $\H_\alpha$. Since $[\Delta,m(1)]=0$, we see that
\begin{equation}
  \label{partitiond}
  \sum_\alpha \H_\alpha = 0 .
\end{equation}
From the partition of unity $\{\varphi_\alpha\}$, we define a series of
differential operators acting on
$\Om^{1/2}(U_{\alpha_0\dots\alpha_k})$:
\begin{equation*}
  \Phi_{\alpha_0\dots\alpha_k} = \frac{\hbar^k}{k+1} \sum_{i=0}^k (-1)^i
  \H_{\alpha_0} \dots \H_{\alpha_{i-1}} m(\varphi_{\alpha_i})
  \H_{\alpha_{i+1}} \dots \H_{\alpha_k} .
\end{equation*}

\begin{lemma}
  \label{eta}
  \begin{equation*}
    [\delta+\hbar\Delta,\Phi_{\alpha_0\dots\alpha_k}]
    = \sum_{i=0}^{k+1} (-1)^i \sum_\alpha
    \Phi_{\alpha_0\dots\alpha_{i-1}\alpha\alpha_i\dots\alpha_k}
  \end{equation*}
\end{lemma}
\begin{proof}
  We have
  \begin{equation*}
    [\delta+\hbar\Delta,\Phi_{\alpha_0\dots\alpha_k}] = \hbar^{k+1} \,
    \H_{\alpha_0} \dots \H_{\alpha_k} .
  \end{equation*}
  On the other hand, we have
  \begin{multline*}
    (-1)^i \sum_\alpha
    \Phi_{\alpha_0\dots\alpha_{i-1}\alpha\alpha_i\dots\alpha_k} \\
    \begin{aligned}
      &= \frac{\hbar^{k+1}}{k+2} \sum_{j=0}^{i-1} (-1)^{j+i}
      \sum_\alpha \H_{\alpha_0} \dots m( \varphi_{\alpha_j} ) \dots
      \H_{\alpha_{i-1}} \H_\alpha \H_{\alpha_i} \dots \H_{\alpha_k} \\
      &+ \frac{\hbar^{k+1}}{k+2} \sum_\alpha \H_{\alpha_0} \dots
      \H_{\alpha_{i-1}} m( \varphi_\alpha ) \H_{\alpha_i} \dots \H_{\alpha_k} \\
      &+ \frac{\hbar^{k+1}}{k+2} \sum_{j=i}^k (-1)^{j+i+1} \sum_\alpha
      \H_{\alpha_0} \dots \H_{\alpha_{i-1}} \H_\alpha \H_{\alpha_i}
      \dots m( \varphi_{\alpha_j} ) \dots \H_{\alpha_k} .
    \end{aligned}
  \end{multline*}
  and by \eqref{partition} and \eqref{partitiond}, this equals
  \begin{equation*}
    \frac{\hbar^{k+1}}{k+2} \, \H_{\alpha_0} \dots \H_{\alpha_k} .
  \end{equation*}
  Summing over $i$, the lemma follows.
\end{proof}

We now come to the main result of this paper.
\begin{theorem}
  \label{T}
  Define a linear form $\T$ on $\Tot\Om^{1/2}(M_\bull)$ by the formula
  \begin{equation*}
    \T(\sigma_\bull) = \sum_{k=0}^\infty (-1)^k
    \sum_{\alpha_0\dots\alpha_k} \int_{\Delta^k}
    \int_{L_{\alpha_0\dots\alpha_k}}
    e^{-\eta_{\alpha_0\dots\alpha_k}/\hbar} \,
    \iota_{\alpha_0\dots\alpha_k}^* \bigl( \Phi_{\alpha_0\dots\alpha_k}
    \sigma_{\alpha_0\dots\alpha_k} \bigr) .
  \end{equation*}
  Then $\T$ is closed: $\T\bigl( (\delta+\hbar\Delta)\sigma_\bull \bigr) = 0$.
\end{theorem}
\begin{proof}
  By Lemma~\ref{eta}, we have
  \begin{multline*}
    \T\bigl( (\delta+\hbar\Delta) \sigma_\bull \bigr)
    \\
      \begin{aligned}
        &= \sum_{k=0}^\infty (-1)^k \sum_{\alpha_0\dots\alpha_k}
        \int_{\Delta^k} \int_{L_{\alpha_0\dots\alpha_k}}
        e^{-\eta_{\alpha_0\dots\alpha_k}/\hbar} \,
        \iota_{\alpha_0\dots\alpha_k}^* \bigl(
        \varphi_{\alpha_0\dots\alpha_k} (\delta+\hbar\Delta)
        \sigma_{\alpha_0\dots\alpha_k} \bigr) \\
        &= \sum_{k=0}^\infty \sum_{\alpha_0\dots\alpha_k} \biggl(
        \int_{\Delta^k} \int_{L_{\alpha_0\dots\alpha_k}}
        e^{-\eta_{\alpha_0\dots\alpha_k}/\hbar} \,
        \iota_{\alpha_0\dots\alpha_k}^* \bigl( (\delta+\hbar\Delta)
        \bigl( \varphi_{\alpha_0\dots\alpha_k}
        \sigma_{\alpha_0\dots\alpha_k}
        \bigr) \bigr) \\
        &- \sum_{i=0}^{k+1} (-1)^i \int_{\Delta^k}
        \int_{L_{\alpha_0\dots\alpha_k}}
        e^{-\eta_{\alpha_0\dots\alpha_k}/\hbar} \,
        \iota_{\alpha_0\dots\alpha_k}^* \bigl(
        \varphi_{\alpha_0\dots\alpha_{i-1}\alpha\alpha_i\dots\alpha_k}
        \sigma_{\alpha_0\dots\alpha_k} \bigr) \biggr) .
    \end{aligned}
  \end{multline*}
  By Theorem~\ref{MS} and Stokes's Theorem, the first sum equals
  \begin{multline*}
    \sum_{k=0}^\infty \sum_{\alpha_0\dots\alpha_k} \int_{\Delta^k}
    \delta \int_{L_{\alpha_0\dots\alpha_k}}
    e^{-\eta_{\alpha_0\dots\alpha_k}/\hbar} \,
    \iota_{\alpha_0\dots\alpha_k}^* \bigl(
    \varphi_{\alpha_0\dots\alpha_k} \sigma_{\alpha_0\dots\alpha_k} \bigr)
    \\
    \shoveleft{
      = \sum_{k=0}^\infty \sum_{i=0}^k (-1)^i
      \sum_{\alpha_0\dots\alpha_k} \int_{\Delta^{k-1}}
    } \\
    \int_{L_{\alpha_0\dots\widehat{\alpha}_i\dots\alpha_k}}
    e^{-\eta_{\alpha_0\dots\widehat{\alpha}_i\dots\alpha_k}/\hbar} \,
    \iota_{\alpha_0\dots\widehat{\alpha}_i\dots\alpha_k}^* \bigl(
    \varphi_{\alpha_0\dots\alpha_k}
    \sigma_{\alpha_0\dots\widehat{\alpha}_i\dots\alpha_k} \bigr) .
  \end{multline*}
  The result follows.
\end{proof}

We may generalize the above construction to families of flexible
Lagrangian submanifolds. In \eqref{flexible}, replace the simplex
$\Delta^k$ by the product $\Delta^k\times\Delta^n$, and allow the
partition of unity $\varphi_\alpha$ to depend on the projection to
$\Delta^n$. Let $\d$ be the de~Rham differential on $\Delta^n$, and
generalize the differential operator $\Phi_{\alpha_0\dots\alpha_k}$ to
have coefficients in $\Om_n$, the differential graded algebra of
differential forms on the auxilliary simplex:
\begin{multline}
  \label{ell}
  \Phi_{\alpha_0\dots\alpha_k} = \frac{1}{k+1} \sum_{i=0}^k (-1)^i
  \bigl( m( \d\varphi_{\alpha_0} ) + \hbar\H_{\alpha_0} ) \dots
  \bigl( m( \d\varphi_{\alpha_{i-1}} ) + \hbar\H_{\alpha_{i-1}} ) \\
  m( \varphi_{\alpha_i} ) \bigl(
  m( \d\varphi_{\alpha_{i+1}} ) + \hbar\H_{\alpha_{i+1}} ) \dots \bigl(
  m( \d\varphi_{\alpha_k} ) + \hbar\H_{\alpha_k} ) .
\end{multline}
With no change in its proof, Lemma \ref{eta} now takes the more
general form
\begin{equation*}
  [\d+\delta+\hbar\Delta,\Phi_{\alpha_0\dots\alpha_k}]
  = \sum_{i=0}^{k+1} (-1)^i \sum_\alpha
  \Phi_{\alpha_0\dots\alpha_{i-1}\alpha\alpha_i\dots\alpha_k} .
\end{equation*}
We may now define a trace $\T$ with values in $\Om_n$, by essentially
the same formula as before, replacing integration over $\Delta^k$ by
integration along the fibers of the projection
$\Delta^k\times\Delta^n\to\Delta^n$. Again with no change in the
proof, Theorem~\ref{T} becomes the following:
\begin{equation*}
  \T\bigl( (\d+\delta+\hbar\Delta)\sigma_\bull \bigr) + \d\T\bigl(
  \sigma_\bull \bigr) = 0
\end{equation*}

An observable in the Batalin--Vilkovisky formalism is a bosonic
half-form $\sigma_\bull$ of ghost number $0$ that is a cocycle in the
Thom-Whitney complex:
\begin{equation*}
  (\d+\delta+\hbar\Delta)\sigma_n=0 .
\end{equation*}
The following is an immediate consequence of this parametrized
generalization of Theorem~\ref{T}.
\begin{corollary}
  If $\sigma_\bull$ is an observable in the Batalin--Vilkovisky
  formalism, then $\T(\sigma_\bull)$ is not changed by an isotopy of
  flexible Lagrangians or a change in the partition of unity used in
  the definition of $\T$.
\end{corollary}

\section{Application to the superparticle}

In this section, we construct a flexible Lagrangian submanifold that
imposes the light-cone gauge for the superparticle. The superparticle
is a supersymmetric analogue of the relativistic particle in
ten-dimensional spacetime $\R^{1,9}$. Recall the action of the
relatavistic particle. (See \cite{superparticle} for further details.)

Let $\{E_\mu\}_{0\le\mu\le9}$ be a basis for $\R^{1,9}$, with inner
product
\begin{equation*}
  ( E_\mu , E_\nu ) =
  \begin{cases}
    \delta_{\mu\nu} , & \mu+\nu>0 , \\
    -1 , & \mu=\nu=0 ,
  \end{cases}
\end{equation*}
with dual basis $\{E^\mu\}_{0\le\mu\le9}$. Let
$\eta^{\mu\nu}=(E^\mu,E^\nu)$.

The world-line is an oriented parametrized one-dimensional manifold,
and fields are differential forms on this manifold. Denote
differentiation along the world-line by $\p$.

The physical fields in this model (the fields of ghost number $0$) are
the position $x^\mu$ and momentum $p_\mu$, which are world-line
scalars taking values in $\R^{1,9}$ and its dual $(\R^{1,9})^\vee$,
and the \emph{einbein} (or gravitational field) $e$, which is a
nowhere-vanishing world-line one-form. The antifields $x^+_\mu$ and
$p^{+\mu}$ are world-line one-forms taking values in $(\R^{1,9})^\vee$
and $\R^{1,9}$ respectively, and the antifield $e^+$ is a world-line
scalar; all have ghost number $-1$ and odd total parity.

There is also a ghost field $c$, associated to parametrization of the
world-line. This field is a world-line scalar of ghost number $1$ and
odd total parity; its antifield $c^+$ is a world-line one-form of
ghost number $-2$ and even total parity.

The Batalin--Vilkovisky action of the relativistic particle is
\begin{equation}
  \label{S}
  S_0 = \tint \left( p_\mu \p x^\mu - \half e (p,p) + \bigl( \p e^+ -
    (x^+,p) \bigr) c \right) \, dt .
\end{equation}

The simply-connected cover $\Spin(1,9)$ of the Lorentz group
$\SO(1,9)$ has a pair of 16-dimensional real representations, the
left- and right-handed Majorana--Weyl spinors $\Ss_\pm$. Denote the
Clifford action of the standard basis of $\R^{1,9}$ by $\gamma^\mu$,
so that
\begin{equation*}
  \gamma^\mu \gamma^\nu + \gamma^\nu \gamma^\mu = 2\eta^{\mu\nu} .
\end{equation*}
The action of the Lie algebra $\so(1,9)$ on $\Ss_\pm$ is realized by
the elements of the Clifford algebra
\begin{equation*}
  \gamma^{\mu\nu} = \half \bigl( \gamma^\mu \gamma^\nu - \gamma^\nu
  \gamma^\mu \bigr) .
\end{equation*}

Denote the non-degenerate pairing between $\Ss_\pm$ and $\Ss_\mp$ by
$\TT(\alpha,\beta)$, and let
\begin{align*}
  \TT^\mu(\alpha,\beta)
  &= \TT(\gamma^\mu\alpha,\beta) = \TT(\alpha,\gamma^\mu\beta)
    \intertext{and}
    \TT^{\mu\nu}(\alpha,\beta)
  &= \TT(\gamma^{\mu\nu}\alpha,\beta) = -\TT(\alpha,\gamma^{\mu\nu}\beta) .
\end{align*}

The superparticle has, in addition to the field content of the
relativistic particle, a sequence $\theta_n$ of world-line scalar
fields of ghost number $n$, which for even $n$ take values in $\Ss_+$
and have odd total parity, and for odd $n$ take values in $\Ss_-$ and
have even total parity. For each $n$ there is the corresponding
antifield $\theta_n^+$, of ghost number $-1-n$, which is a world-line
one-form that for even $n$ takes values in $\Ss_-$ and has odd total
parity, and for odd $n$ takes values in $\Ss_-$ and has even total
parity. Let $N$ be the graded supermanifold with coordinates
$\{x^\mu,p_\mu,e,c,\theta_n\}$, and let $M=T^*[-1]N$ be its shifted
cotangent bundle, whose fibres have coordinates the corresponding
antifields.

The Batalin--Vilkovisky extension of the classical action of the
superparticle has the form $S=S_0+S_1$, where $S_0$ is the classical
action of the particle \eqref{S}, and $S_1$ depends only on the fields
and antifields
\begin{equation*}
  \{p_\mu,\theta_n\}\cup\{x^+_\mu,e^+,c^+,\theta^+_n\}
\end{equation*}
and their derivatives. The salient term
$-\half \int p_\mu \TT^\mu(\theta_0,\p\theta_0) \, dt$ of $S_1$ is the
dimensional reduction of the topological term in the Green--Schwarz
action for the superstring.

A light-like vector $\m$ is a non-zero vector such that
\begin{equation*}
  (\m,\m) = 0 .
\end{equation*}
The space of solutions of this equation has two components, the
forward and backward light-cones, depending on the sign of the
time-component $\m_0$. Fix such a light-like vector, for example
$\m=\half(E^0+E^9)$.

Let $\n$ be another light-like vector satisfying the equation
$(\m,\n)=1/2$, for example $\n=\half(E^0-E^9)$. Let $\cl(\m)$ be
Clifford multiplication by $\m$, given by contraction with the
$\gamma$-matrices,
\begin{equation*}
  \cl(\m) = \m_\mu\gamma^\mu : \Ss_\pm \to \Ss_\mp ,
\end{equation*}
and similarly for $\cl(\n)$. We have $\cl(\m)^2=\cl(\n)^2=0$, and
$\cl(\m)\cl(\n)+\cl(\n)\cl(\m)=1$.

The orthogonal complement to the plane spanned by the vectors
$\{\m,\n\}$ is the eight-dimensional Euclidean space spanned by
$\{E_1,\dots,E_8\}$. Associated to this subspace are a pair of spinor
representations $\ss_\pm$, also eight-dimensional, which we may
identify with solutions in $\Ss_\pm$ of the equation
\begin{equation*}
  \cl(\m)\theta_\pm = 0 ,
\end{equation*}
The map $\theta\mapsto\cl(\m)\cl(\n)\theta$ projects from $\Ss_\pm$ to
$\ss_\pm$.
\begin{lemma}
  \label{lightcone}
  If $\theta_\pm\in\ss_\pm$, then
  $p_\mu\TT^\mu(\theta_+,\theta_-) = 2 (p,\m) \, \n_\nu
  \TT^\nu(\theta_+,\theta_-)$.
\end{lemma}
\begin{proof}
  Let $q=p-2\,(p,\n)\m-2\,(p,\m)\n$. We have
  \begin{align*}
    \cl(p)
    &= \cl(p)\cl(\m)\cl(\n) + \cl(p)\cl(\n)\cl(\m) \\
    &= \bigl( - \cl(\m) \cl(q) \cl(\n) + 2\, (p,\m)
      \cl(\n)\cl(\m)\cl(\n) \bigr) \\
    &+ \bigl( - \cl(\n) \cl(q) \cl(\m) + 2\, (p,\n)
      \cl(\m)\cl(\n)\cl(\m) \bigr) \\
    &= - \cl(\m) \cl(q) \cl(\n) - \cl(\n) \cl(q) \cl(\m) +
      2\, (p,\m)\cl(\n) + 2\, (p,\n)\cl(\m) .
  \end{align*}
  It follows that
  \begin{align*}
    p_\mu \TT^\mu(\theta_+,\theta_-)
    &= - q_\mu \bigl( \TT^\mu(\cl(\n)\theta_+,\cl(\m)\theta_-)
      + \TT^\mu(\cl(\m)\theta_+,\cl(\n)\theta_-) \bigr) \\
    &+ 2\, (p,\m) \TT(\cl(\n)\theta_+,\theta_-)
      + 2\, (p,\n) \TT(\cl(\m)\theta_+,\theta_-) \\
    &= 2\, (p,\m) \TT(\cl(\n)\theta_+,\theta_-) .
      \qedhere
  \end{align*}
\end{proof}

Let $U(\m)$ be the open subset of $M$ where the inequality $(p,\m)>0$
holds. Let $L(\m)$ be the Lagrangian submanifold in $U(\m)$ defined by
the equations
\begin{align*}
  x^+_\mu &= 0 , & p^{+\mu} &= 0 , & e &= 1 , &
 c^+ &= 0 , & \cl(\m)\theta_n &= 0 , & \cl(\m)\theta^+_n &= 0 ,
\end{align*}
and let $\iota$ be the inclusion of $L(\m)$ in $U(\m)$.

The explicit formula for the Batalin--Vilkovisky extension of the
classical action of the superparticle and Lemma \ref{lightcone} give
the following proposition.
\begin{lemma}
  \label{fixed}
  The gauge-fixed action $\iota^*S$ for the superparticle equals
  \begin{equation*}
    \int \Bigl( p_\mu \p x^\mu - \half (p,p) + \p e^+ c + (p,\m)
    \, \n_\mu \Bigl( - \TT^\mu(\theta_0,\p\theta_0) + 2
    \sum_{n=0}^\infty
    \TT^\mu(\theta^+_n,\theta_{n+1}) \Bigr) \Bigr) \, dt .
  \end{equation*}
  In particular, there is no dependence on the function $\Phi(p)$ at
  the classical level.
\end{lemma}

In terms of the redefined spinor fields
\begin{align}
  \label{redefine}
  \Theta_n
  &= (p,\m)^{n+1/2} \theta_n , &
  \Theta_n^+
  &= (p,\m)^{-n-1/2} \theta_n^+ ,
\end{align}
this gauge-fixed theory has the same action as a free massless
relativistic superparticle, with classical action
\begin{equation*} 
  \int \Bigl( p_\mu \p x^\mu - \half (p,p) + \p e^+ c - \n_\mu
  \TT^\mu(\Theta_0,\p\Theta_0) + 2 \sum_{n=0}^\infty \n_\mu
  \TT^\mu(\Theta^+_n,\Theta_{n+1}) \Bigr) \, dt .
\end{equation*}

At least after regularization, the Berezinian of the canonical
transformation \eqref{redefine} is seen to be proportional to the
value of the Dirichlet L-function
\begin{equation*}
  L(s) = 1^{-s} - 3^{-s} + 5^{-s} - 7^{-s} + 9^{-s} - \dots
\end{equation*}
at $s=-1$. As Hurwitz showed \cite{Hurwitz},
\begin{equation*}
  L(-1) = 1 - 3 + 5 - 7 + 9 - \dots =
  \sum_{\substack{n\in\Z\\n\equiv1\bmod{4}}} n
\end{equation*}
vanishes. In terms of the Hurwitz zeta-function
\begin{equation*}
  \zeta(s,a) = a^{-s} + (a+1)^{-s} + (a+2)^{-s} + (a+3)^{-s} + \dots ,
\end{equation*}
we clearly have
\begin{equation*}
  L(s) = 4^{-s} \left( \zeta(s,1/4) - \zeta(s,3/4) \right) .
\end{equation*}
The values of the Hurwitz zeta-function at negative integers is
related to the Bernoulli polynomials $B_n(a)$, by the formula
\begin{equation*}
  \zeta(1-n,a) = - B_n(a)/n .
\end{equation*}
(This is a consequence of the equation
$\p\zeta(s,a)/\p s=-s\zeta(s+1,a)$ and Euler's formulas
$\zeta(1-n)=(-1)^nB_n/n$, $n>0$, and $\text{Res}_{s=1}\zeta(s)=1$.)
Since $B_n(1-a)=(-1)^nB_n(a)$, and in particular $B_2(a)=a^2-a+1/6$,
the vanishing of $L(-1)$ follows. This L-function was first presented
by Schl\"omilch in the guise of a problem for university students
\cite{Schlomilch}, and we reproduce here the original text.

\begin{figure}[h]
  \centering
  \fbox{\includegraphics[height=3in]{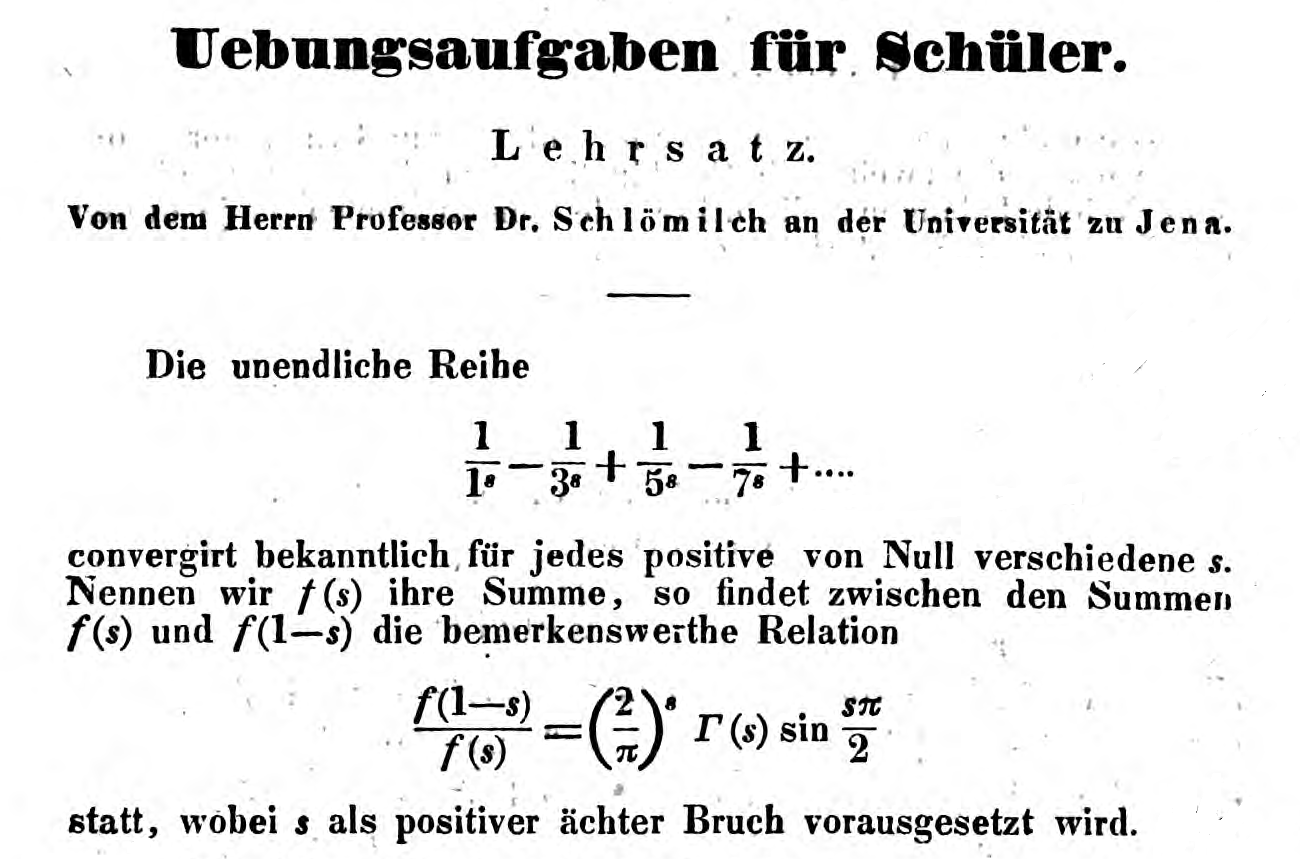}}
\end{figure}

Arguably the physics of the superparticle takes place in the vicinity
of the forward light-cone $\{(p,p)=0\}\cap\{p_0>0\}$. The open set
$U(\m)$ does not cover the whole of the forward light-cone: in fact,
it omits the ray where $p_0=p_9$ and the transverse momenta vanish.
This is a symptom of a Gribov ambiguity in the use of the light-cone
gauge, which appears to have received little attention in the
literature.\footnote{An exception is Siegel \cite{Siegel}. He resolves
  the problem in the absence of supersymmetry by a method akin to
  stochastic quantization.  He adjoins additional coordinates in the
  target spacetime: a pair of bosonic fields, with metric of signature
  $(1,1)$, and a pair of compensating fermionic fields. The effect is
  to replace the Lorentz group $\SO(1,9)$ by the orthosymplectic
  supergroup $\SOSp(2,9|2)$. Extending this method from the particle
  to the superparticle would seem to require the introduction of
  spinors for the superspace $\R^{1,1|2}$, or equivalently,
  differential forms on an auxilliary line. We will not pursue this
  approach further here.}

The open set
\begin{equation*}
  U = \{ p_1^2 + \dots + p_9^2 > \half p_0^2 \} \cap \{ p_0 > 0 \}
\end{equation*}
is a neighbourhood of the forward light-cone. Consider the light-like
vectors
\begin{equation*}
  \m_\pm = \half (E^0\pm E^9) ,
\end{equation*}
and the associated open sets $U(\m_\pm)=\{p_0>\pm p_9\}$. We consider
the subsets
\begin{equation*}
  U_\pm = \{ p_0 > \pm 2\,p_9 \} \cap U \subset U(\m_\pm) .
\end{equation*}
The open sets $U_+$ and $U_-$ cover $U$, and have intersection
\begin{equation*}
  U_{+-} = U_+ \cap U_- = \{ p_1^2 + \dots + p_9^2 > \half p_0^2 > 2\,p_9^2 \} .
\end{equation*}

We adopt the convention that indices in the range $\{1,\dots,8\}$,
that is, transverse to the light-cone, are denoted $p_a$, $p^{+b}$,
etc. The Einstein summation convention will also be applied for these
indices. Consider the Hamiltonian flow $\Phi_\tau$ associated to the
Hamiltonian
\begin{equation}
  \label{Phi}
  \psi = - \frac{\pi}{p_*}
  \sum_{n=0}^\infty p_a \TT^{a 9}(\theta^+_n,\theta_n) ,
\end{equation}
where
$p_*=\bigl(\eta^{ab}p_ap_b\bigr)^{1/2} =
\bigl(p_1^2+\dots+p_8^2\bigr)^{1/2}$. This flow leaves all of the
fields invariant except $p^{+a}$, $\theta_n$ and $\theta^+_n$. In
terms of the one-parameter group in $\Spin(1,9)$,
\begin{equation*}
  g(\tau) = \cos(\pi\tau/2) - \frac{\sin(\pi\tau/2)}{p_*} p_a
  \gamma^{a9} ,
\end{equation*}
these fields transform as follows: $\theta_n(\tau)= g(\tau)\theta_n$,
$\theta^+_n(\tau)=g(\tau)\theta^+_n$, and
\begin{equation*}
  p^{+a}(\tau) = p^{+a} + \sum_{n=0}^\infty \TT\biggl(
  g(\tau)^{-1}\frac{\p g(\tau)}{\p p_a} \, \theta^+_n,\theta_n \biggr) .
\end{equation*}
The formulas for $\theta_n(\tau)$ and $\theta^+_n(\tau)$ are clear,
and the formula for $p^{+a}$ on solving for the antibrackets
$\bigl(p_a,p^{+b}(\tau)\bigr)=\delta^b_a$,
$\bigl( \theta_n(\tau) , p^{+a}(\tau) \bigr)=0$ and
$\bigl( \theta^+_n(\tau) , p^{+a}(\tau) \bigr)=0$.

Consider the Lagrangian submanifolds
\begin{equation*}
  L_\pm = L(\m_\pm)\cap U_\pm ,
\end{equation*}
with inclusions $\iota_\pm:L_\pm\to U_\pm$. Let $L_{+-}=L_+\cap U_-$,
and let $L(\tau)$ be the image of $L_{+-}$ under the canonical
transformation $\Phi_\tau$. A calculation in the Clifford algebra
shows that
\begin{equation*}
  g(\tau)\cl(\m)g(\tau)^{-1} = \cl(\m(\tau)) ,
\end{equation*}
where $\m(\tau)$ is the light-like vector
\begin{equation*}
  \m(\tau) = \frac{1}{2} \biggl( E^0 + \cos(\pi\tau) E^9 -
  \frac{\sin(\pi\tau)}{p_*} p_a E^a \biggr) . 
\end{equation*}
On $L(\tau)$, we have $\cl(\m(\tau)) \theta_n(\tau)= 0$ and
$\cl(\m(\tau)) \theta^+_n(\tau)= 0$, and hence
$\theta_n=\cl(\m(\tau))\cl(\n(\tau))$, where $\n(\tau)$ is the
light-like vector
\begin{equation*}
  \n(\tau) = \frac{1}{2} \biggl( E^0 - \cos(\pi\tau) E^9 +
  \frac{\sin(\pi\tau)}{p_*} p_a E^a \biggr) . 
\end{equation*}
The Lagrangian $L(\tau)$ is cut out by the equations $x^+_\mu=0$,
$p^{+\mu}(\tau)=0$, $e=1$, $c^+=0$, $\cl(\m(\tau))\theta_n=0$ and
$\cl(\m(\tau))\theta^+_n=0$.

\begin{lemma}
  \label{p+}
  On $L(\tau)$, we have
  \begin{equation*}
    p^{+a}(\tau) = p^{+a} + \frac{\sin(\pi\tau)}{2p_*}  \biggl(
    \frac{p_b \eta^{ab}}{(p_*)^2} p_c \TT^{c9}(\theta^+_n,\theta_n) -
    \TT^{a9}(\theta^+_n,\theta_n) \biggr) .
  \end{equation*}
\end{lemma}
\begin{proof}
  Since $\theta_n(\tau)=\cl(\m(\tau))\cl(\n(\tau))\theta_n(\tau)$, we
  see that
  \begin{align*}
    p^{+a}(\tau)
    &= p^{+a} + \sum_{n=0}^\infty \TT\biggl(
      \frac{\p g(\tau)}{\p p_a} g(\tau)^{-1} \, \theta^+_n(\tau)
      , \theta_n(\tau) \biggr) \\
    &= p^{+a} + \sum_{n=0}^\infty \TT\biggl(
      \frac{\p g(\tau)}{\p p_a} g(\tau)^{-1} \, \theta^+_n(\tau) ,
      \cl(\m(\tau))\cl(\n(\tau))\theta_n(\tau) \biggr) \\
    &= p^{+a} + \sum_{n=0}^\infty \TT\biggl(
      g(\tau)\cl(\n)\cl(\m)g(\tau)^{-1}\frac{\p g(\tau)}{\p p_a}
      g(\tau)^{-1} \, \theta^+_n(\tau) , \theta_n(\tau) \biggr) .
  \end{align*}
  Observe that
  \begin{equation*}
    \cl(\n) \cl(\m) g(\tau)^{-1} = \cos(\pi\tau/2) \cl(\n) \cl(\m) +
    \frac{\sin(\pi\tau/2)}{p_*} p_a \gamma^{a9}
    \cl(\m) \cl(\n) .
  \end{equation*}
  Since
  \begin{equation*}
    \frac{\p g(\tau)}{\p p_a} = \frac{\sin(\pi\tau/2)}{p_*}  \biggl(
    \frac{p_b \eta^{ab}}{(p_*)^2} p_c \gamma^{c9} - \gamma^{a9}
    \biggr) ,
  \end{equation*}
  we see that
  \begin{align*}
    \cl(\m) \cl(\n) \frac{\p g(\tau)}{\p p_a} g(\tau)^{-1}
    \theta^+_n(\tau)
    &= \frac{\p g(\tau)}{\p p_a} \cl(\n) \cl(\m)
      g(\tau)^{-1} \theta^+_n(\tau) \\
    &= \frac{\p g(\tau)}{\p p_a} g(\tau)^{-1} \cl(\n(\tau))
      \cl(\m(\tau)) \theta^+_n(\tau) \\
    &= 0 \intertext{and}
      \cl(\n) \cl(\m) \frac{\p g(\tau)}{\p p_a} g(\tau)^{-1}
      \theta^+_n(\tau)
    &= \frac{\p g(\tau)}{\p p_a} \cl(\m) \cl(\n)
      g(\tau)^{-1} \theta^+_n(\tau) \\
    &= \frac{\p g(\tau)}{\p p_a} g(\tau)^{-1} \cl(\m(\tau))
      \cl(\n(\tau)) \theta^+_n(\tau) \\
    &= \frac{\p g(\tau)}{\p p_a} g(\tau)^{-1} \theta^+_n(\tau) .
  \end{align*}
  The lemma follows.
\end{proof}

Note that $\m(0)=\m_+$ and $\m(1)=\m_-$, while $\n(0)=\m_-$ and
$\n(1)=\m_+$. In conjunction with the lemma, this implies the
following corollary.
\begin{corollary}
  $L(1)=L_-\cap U_{+-}$
\end{corollary}

Let $\iota_{+-}: L_{+-}\times\Delta^1\to U_{+-}$ be the family of
Lagrangians equal to $L(\tau)$ at $\tau\in\Delta^1$. Together with the
Lagrangians $L_\pm$, we obtain a flexible Lagrangian for the cover
$\{U_+,U_-\}$.

By Lemma \ref{fixed}, we have
\begin{multline*}
  \iota_{+-}^*S = \\
  \int \Bigl( p_\mu \p x^\mu - \half (p,p) + \p e^+ c + p(\tau) \,
  \n_\mu(\tau) \Bigl( - \TT^\mu(\theta_0,\p\theta_0) + 2
  \sum_{n=0}^\infty \TT^\mu(\theta^+_n,\theta_{n+1}) \Bigr) \Bigr) \,
  dt ,
\end{multline*}
where $p(\tau)=(p,\m(\tau))$. Here, we have used that the action $S$
of the superparticle does not depend on $p^{+\mu}$.
\begin{lemma}
  The functions $p(\tau)$, $0\le\tau\le1$,and $p_*$ are positive on
  $U_{+-}$.
\end{lemma}
\begin{proof}
  On $U_{+-}$, we have
  \begin{equation*}
    p(\tau) = \half \bigl( p_0 - \cos(\pi\tau) p_9 + \sin(\pi\tau) p_*
    \bigr) \ge \half ( p_0 - |p_9| ) > 0 .
  \end{equation*}
  Likewise,
  \begin{equation*}
    \half p_0^2 < p_*^2+p_9^2 < p_*^2 + \tfrac{1}{4} p_0^2 ,
  \end{equation*}
  and hence $p_*^2>\frac{1}{4} p_0^2$, showing that $p_*$ is positive
  on $U_{+-}$.
\end{proof}

We now perform the change of variables \eqref{redefine}:
\begin{align*}
  \Theta_n
  &= p(\tau)^{n+1/2} \theta_n , &
  \Theta_n^+
  &= p(\tau)^{-n-1/2} \theta_n^+ .
\end{align*}
The resulting gauge-fixed action $\iota_{+-}^*S$ is independent of the
parameter $\tau$. The one-form $\eta_{+-}=\psi\,d\tau$ in the
contribution of $U_{+-}$ to the functional integral $\T(\sigma_\bull)$
equals
\begin{equation*}
  \eta_{+-} = - \frac{\pi}{p_*} \sum_{n=0}^\infty p_a
  \TT^{a9}(\Theta^+_n,\Theta_n) \, d\tau .
\end{equation*}

To complete the formula for $\T(\sigma_\bull)$, we need a partition of
unity for the cover $U_\pm$ of $U$. Choose a function
$\varphi\in C^\infty(\R)$ that vanishes for $t\le\tfrac14$ and such
that $\varphi(t)+\varphi(1-t)=1$. A suitable partition of unity is
\begin{equation*}
  \varphi_\pm(p) = \varphi\bigl( (p_9\mp p_0)/2p_9 \bigr) .
\end{equation*}

\section{Global symmetries}

The superparticle is Lorentz invariant and supersymmetric. On the
other hand, the flexible Lagrangian that we constructed in the last
section is not invariant under these symmetries. In this section, we
give an equivariant extension of Theorem~\ref{T}. There is now a
Hamiltonian action of a (finite-dimensional) Lie superalgebra $\g$ on
the Batalin--Vilkovisky supermanifold $M$. We do not assume any
compatibility between this action and either the cover $\CU$ or the
simplicial Lagrangian $L_\bull$. Instead, we express the covariance of
the linear form $\T$ by adapting the BRST formalism. In practice, this
means that we replace the complex numbers by the commutative
superalgebra $C^*(\g)$ of cochains of the Lie superalgebra $\g$.

The action of $\g$ on $M$ is determined by a moment map, that is, a
morphism of Lie superalgebras $\rho:\g\to\CO(M)[-1]$. In other words,
if $\xi_1,\xi_2\in\g$, we have
\begin{equation*}
  (\rho(\xi_1),\rho(\xi_2)) = \rho([\xi_1,\xi_2]) .
\end{equation*}

We now introduce the differential graded commutative superalgebra
$C^*(\g)$ of Lie superalgebra cochains on $\g$. This is the free
graded commutative superalgebra generated by the dual superspace
$\g^\vee[-1]$ to $\g$ placed at ghost number $1$. If $\{\xi_a\}$ is a
(homogeneous) basis of $\g$, then $C^*(\g)$ is generated by elements
$\{\epsilon^a\}$ of ghost number $1$, having the opposite total degree
to $\xi_a$: if $\xi_a$ is even (respectively odd), $\epsilon^a$ is an
exterior (resp.\ polynomial) generator.

The structure coefficients of $\g$ are defined as follows:
\begin{equation*}
  [\xi_a,\xi_b] = C_{ab}^c \, \xi_c .
\end{equation*}
The differential on $C^*(\g)$ is given by the formula
\begin{equation*}
  \delta_\g \epsilon^a = \half \sum_{b,c} (-1)^{(\pa(\xi_b)+1)\pa(\xi_c)}
  \, C_{bc}^a \, \epsilon^b\epsilon^c .
\end{equation*}
The element
\begin{equation*}
  \mu = \sum_a \rho(\xi_a) \epsilon^a \in \CO(M) \o C^*(\g)
\end{equation*}
satisfies the Maurer-Cartan equation:
\begin{equation*}
  \delta_\g\mu + \half \bigl( \mu , \mu \bigr) = 0 .
\end{equation*}
This implies the following identity for differential operators on
$\Om^{1/2}(M)\o C^*(\g)$:
\begin{equation*}
  e^{\mu/\hbar} \circ ( \delta_\g + \H_\mu + \hbar\Delta ) = 
  ( \delta_\g + \hbar\Delta ) \circ e^{\mu/\hbar} .
\end{equation*}
In particular, $\delta_\g + \H_\mu + \hbar\Delta$ is a differential on
$\Om^{1/2}(M)\o C^*(\g)$.

We have the following equivariant extension of Theorem~\ref{T}.
\begin{theorem}
  \label{Tg}
  Define a linear form $\T_\g$ on
  $\Tot\CO(M_\bull)\o\Om_\ell\o C^*(\g)$ with values in
  $\Om_\ell\o C^*(\g)$ by the formula
  \begin{equation*}
    \T_\g(\sigma_\bull) = \sum_{k=0}^\infty (-1)^k
    \sum_{\alpha_0\dots\alpha_k} \int_{\Delta^k}
    \int_{L_{\alpha_0\dots\alpha_k}}
    e^{-\eta_{\alpha_0\dots\alpha_k}/\hbar} \,
    \iota_{\alpha_0\dots\alpha_k}^* \bigl( \psi_{\alpha_0\dots\alpha_k}
    \bigl(
    e^{\mu/\hbar} \sigma_{\alpha_0\dots\alpha_k} \bigr) \bigr) .
  \end{equation*}
  Then $\T_\g$ is closed:
  $\T_\g\bigl( (\d+\delta_\g+\H_\mu+\delta+\hbar\Delta)\sigma_\bull
  \bigr) + (\d+\delta_\g)\T_\g\bigl( \sigma_\bull \bigr) = 0$.
\end{theorem}

We apply this theorem to the superparticle. The Lie superalgebra $\g$
is the sum of three subspaces: translations, parametrized by a
covariant vector in $\R^{1,9}$, supersymmetries, parametrized by a
Majorana--Weyl spinor in $\Ss_-$, and Lorentz transformations,
parametrized by $\so(1,9)$, or equivalently, the second exterior power
$\Lambda^2\R^{1,9}$. The momentum for translation symmetry equals
\begin{equation*}
  \int x^+_\mu \, dt .
\end{equation*}
The restriction of $x^+_\mu$ to the flexible Lagrangian of the
previous section vanishes. Thus, translations may be ignored in the
calculation of $\mu$.

The momentum for supersymmetry equals
\begin{equation*}
  \int \bigl( \theta^+_0 - \half x^+_\mu \gamma^\mu \theta_0 \bigr) \,
  dt .
\end{equation*}
On restriction to the flexible Lagrangian by any of the maps
$\iota_\pm$ or $\iota_{+-}$, the second term vanishes.  Denoting the
corresponding BRST ghosts, of ghost number $1$ and even total parity,
by $\epsilon\in\Ss_+$, we obtain a contribution of
$\int \TT(\theta^+_0,\epsilon) \, dt$ to $\mu$ in all three cases.

The momentum for Lorentz symmetries is
\begin{equation*}
  \int \Bigl( \eta^{\lambda[\mu} x^{\nu]} x^+_\lambda
  - \eta^{\lambda[\mu} p^{+\nu]} p_\lambda
  - \sum_{n=0}^\infty\TT^{\mu\nu}(\theta^+_n,\theta_n) \Bigr) \, dt .
\end{equation*}
Its contribution to $\iota_+^*\mu$ and $\iota_-^*\mu$ equals
\begin{equation*}
  - \sum_{n=0}^\infty \int \bigl(
  \TT^{a0}(\theta^+_n,\theta_n) \, \epsilon_{a0}
  + \TT^{a9}(\theta^+_n,\theta_n) \, \epsilon_{a9} \bigr) \, dt ,
\end{equation*}
since $x^+_\mu$ and $p^{+\mu}$ vanish on $L_+$ and $L_-$, and
$\gamma^{ab}$ and $\gamma^{09}$ commute with $\cl(\m_+)$ and
$\cl(\m_-)$. Its contribution to $\iota^*_{+-}\mu$ may be derived from
the formula of Lemma~\ref{p+} for $p^{+\mu}(\tau)$.


\begin{bibdiv}
\bibliographystyle{JHEP}
\begin{biblist}

\bib{BGI}{article}{
   author={Becchi, Carlo},
   author={Giusto, Stefano},
   author={Imbimbo, Camillo},
   title={The functional measure of gauge theories in the presence of Gribov
   horizons},
   conference={
      title={Path integrals from peV to TeV},
      address={Florence},
      date={1998},
   },
   book={
      publisher={World Sci. Publ., River Edge, NJ},
   },
   date={1999},
   pages={36--43},
   review={\MR{1726570}},
}

\bib{BF}{article}{
   author={Behrend, Kai},
   author={Fantechi, Barbara},
   title={Gerstenhaber and Batalin--Vilkovisky structures on Lagrangian
   intersections},
   conference={
      title={Algebra, arithmetic, and geometry: in honor of Yu. I. Manin.
      Vol. I},
   },
   book={
      series={Progr. Math.},
      volume={269},
      publisher={Birkh\"{a}user Boston, Inc., Boston, MA},
   },
   date={2009},
   pages={1--47},
   review={\MR{2641169}},
   doi={10.1007/978-0-8176-4745-2\_1},
}

\bib{BottTu}{book}{
   author={Bott, Raoul},
   author={Tu, Loring W.},
   title={Differential forms in algebraic topology},
   series={Graduate Texts in Mathematics},
   volume={82},
   publisher={Springer-Verlag, New York-Berlin},
   date={1982},
   pages={xiv+331},
   isbn={0-387-90613-4},
   review={\MR{658304}},
}

\bib{Costello}{book}{
   author={Costello, Kevin},
   title={Renormalization and effective field theory},
   series={Mathematical Surveys and Monographs},
   volume={170},
   publisher={American Mathematical Society, Providence, RI},
   date={2011},
   pages={viii+251},
   isbn={978-0-8218-5288-0},
   review={\MR{2778558}},
   doi={10.1090/surv/170},
}

\bib{covariant}{article}{
   author={Getzler, Ezra},
   title={Covariance in the Batalin-Vilkovisky formalism and the
   Maurer-Cartan equation for curved Lie algebras},
   journal={Lett. Math. Phys.},
   volume={109},
   date={2019},
   number={1},
   pages={187--224},
   issn={0377-9017},
   review={\MR{3897596}},
   doi={10.1007/s11005-018-1106-8},
}

\bib{superparticle}{article}{
  author={Getzler, Ezra},
  author={Pohorence, Sean Weinz},
  journal={Adv. Math. Theor. Phys.},
  date={2019},
  volume={23},
  number={6},
  title={Covariance of the classical Brink--Schwarz superparticle},
}  

\bib{Gribov}{article}{
   author={Gribov, V. N.},
   title={Quantization of non-Abelian gauge theories},
   journal={Nuclear Phys. B},
   volume={139},
   date={1978},
   number={1-2},
   pages={1--19},
   issn={0550-3213},
   review={\MR{0496136}},
   doi={10.1016/0550-3213(78)90175-X},
}

\bib{Hurwitz}{article}{
  author={Hurwitz, Adolf},
  title={Einige Eigenschaften der Dirichlet'schen Funktionen
  $F(s)=\sum \left( \frac{D}{n} \right) \cdot \frac{1}{n^s}$, die
  bei der Bestimmung der Klassen-anzahlen bin\"arer quadratischer
  Formen auftreten},
  journal={Zeitschrift f\"ur Mathematik und Physik},
  volume={27},
  date={1882},
  pages={86--101}
}

\bib{Khudaverdian}{article}{
   author={Khudaverdian, Hovhannes M.},
   title={Semidensities on odd symplectic supermanifolds},
   journal={Comm. Math. Phys.},
   volume={247},
   date={2004},
   number={2},
   pages={353--390},
   issn={0010-3616},
   review={\MR{2063265}},
   doi={10.1007/s00220-004-1083-x},
}

\bib{KV}{article}{
   author={Khudaverdian, Hovhannes M.},
   author={Voronov, Theodore Th.},
   title={Differential forms and odd symplectic geometry},
   conference={
      title={Geometry, topology, and mathematical physics},
   },
   book={
      series={Amer. Math. Soc. Transl. Ser. 2},
      volume={224},
      publisher={Amer. Math. Soc., Providence, RI},
   },
   date={2008},
   pages={159--171},
   review={\MR{2462360}},
   doi={10.1090/trans2/224/08},
}

\bib{Leites}{article}{
  author={Leites, D. A.},
  title={Lie superalgebras},
  journal={Journal of Soviet Mathematics},
  year={1985},
  volume={30},
  number={6},
  pages={2481--2512},
  issn={1573-8795},
  doi={10.1007/BF02249121},
  url={https://doi.org/10.1007/BF02249121}
}

\bib{Manin}{book}{
   author={Manin, Yuri I.},
   title={Gauge field theory and complex geometry},
   series={Grundlehren der Mathematischen Wissenschaften
   },
   volume={289},
   note={Translated from the Russian by N. Koblitz and J. R. King},
   publisher={Springer-Verlag, Berlin},
   date={1988},
   pages={x+297},
   isbn={3-540-18275-6},
   review={\MR{954833}},
}

\bib{Schlomilch}{article}{
  author={Schl\"{o}milch, Oskar},
  journal={Grunert's Archiv der Math.\ u.\ Physik, Ser.\ I},
  number={12},
  date={1849},
  pages={415}
}

\bib{MS}{article}{
   author={Mikhailov, Andrei},
   author={Schwarz, Albert},
   title={Families of gauge conditions in BV formalism},
   journal={J. High Energy Phys.},
   date={2017},
   number={7},
   pages={063, front matter+24},
   issn={1126-6708},
   review={\MR{3686735}},
   doi={10.1007/JHEP07(2017)063},
}

\bib{Schwarz}{article}{
   author={Schwarz, Albert},
   title={Geometry of Batalin--Vilkovisky quantization},
   journal={Comm. Math. Phys.},
   volume={155},
   date={1993},
   number={2},
   pages={249--260},
   issn={0010-3616},
   review={\MR{1230027}},
}

\bib{Severa}{article}{
   author={\v{S}evera, Pavol},
   title={On the origin of the BV operator on odd symplectic supermanifolds},
   journal={Lett. Math. Phys.},
   volume={78},
   date={2006},
   number={1},
   pages={55--59},
   issn={0377-9017},
   review={\MR{2271128}},
   doi={10.1007/s11005-006-0097-z},
}

\bib{Siegel}{article}{
   author={Siegel, W.},
   title={Covariantly second-quantized string},
   journal={Phys. Lett. B},
   volume={142},
   date={1984},
   number={4},
   pages={276--280},
   issn={0370-2693},
   review={\MR{756203}},
   doi={10.1016/0370-2693(84)91197-3},
}

\bib{Singer}{article}{
   author={Singer, I. M.},
   title={Some remarks on the Gribov ambiguity},
   journal={Comm. Math. Phys.},
   volume={60},
   date={1978},
   number={1},
   pages={7--12},
   issn={0010-3616},
   review={\MR{500248}},
}

\bib{Sullivan}{article}{
   author={Sullivan, Dennis},
   title={Infinitesimal computations in topology},
   journal={Inst. Hautes \'Etudes Sci. Publ. Math.},
   number={47},
   date={1977},
   pages={269--331},
   issn={0073-8301},
   review={\MR{0646078}},
}

\bib{Weinstein}{article}{
   author={Weinstein, Alan},
   title={Symplectic manifolds and their Lagrangian submanifolds},
   journal={Advances in Math.},
   volume={6},
   date={1971},
   pages={329--346},
   issn={0001-8708},
   review={\MR{0286137}},
   doi={10.1016/0001-8708(71)90020-X},
}

\end{biblist}
\end{bibdiv}
\end{document}